\documentclass{amsart}

\usepackage{amssymb}
\usepackage{amsmath}
\usepackage{amsthm}
\usepackage{graphicx}
\usepackage{psfrag}

\title{Hardness and Parameterized Algorithms on Rainbow Connectivity problem}

\author{Prabhanjan Ananth} 
\address{Dept. of Computer Science and Automation, Indian Institute of Science.}
\email{prabhanjan@csa.iisc.ernet.in}

\author{Meghana Nasre}
\address{Dept. of Computer Science and Automation, Indian Institute of Science.}
\email{meghana@csa.iisc.ernet.in}

\author{Kanthi K Sarpatwar}
\address{Computer Science Department, University of Maryland, College Park, USA.}
\email{kanthik@gmail.com}

\newtheorem{theorem}{Theorem}
\newtheorem{lemma}[theorem]{Lemma}
\newtheorem{corollary}[theorem]{Corollary}
\newcommand {\coloring}{\chi}

\begin{document}
\maketitle
\begin{abstract}
A path in an edge colored graph is said
to be a rainbow path if no two edges on
the path have the same color. An edge colored
graph is (strongly) rainbow connected if there
exists a (geodesic) rainbow path between every pair
of vertices. The (strong) rainbow connectivity
of a graph $G$, denoted by ($src(G)$, respectively)
$rc(G)$ is the smallest number of colors
required to edge color the graph such that
$G$ is (strongly) rainbow connected. 
In this paper we study the rainbow connectivity problem and the strong rainbow
connectivity problem from a computational point of view.
Our main results can be summarised as below:
\begin{itemize}
\item For every fixed $k \geq 3$, it is NP-Complete to decide whether
$src(G) \leq k$ even when the graph $G$ is bipartite.
\item For every fixed odd $k \geq 3$, it is NP-Complete to decide whether $rc(G) \leq k$.
This resolves one of the open problems posed by Chakraborty et al. (J. Comb. Opt., 2011) where they prove the
hardness for the even case.
\item The following problem is {\em fixed parameter tractable}: Given a graph $G$, determine
the maximum number of pairs of vertices that can be rainbow connected using two colors.
\item For a directed graph $G$, it is NP-Complete to decide whether $rc(G) \leq 2$.
\end{itemize}
\end{abstract}

\section{Introduction}
This paper deals with the notion of {\em rainbow connectivity} and {\em strong rainbow connectivity}
of a graph.
Unless mentioned otherwise, all the graphs are assumed to be connected and undirected. 
Consider an edge coloring (not necessarily proper) of a graph $G = (V, E)$.  
A path between a pair of vertices is said to be a 
{\em rainbow path}, if no two edges
on the path have the same color.
If the edges of $G$ can be colored using $k$ colors such that, between every pair of vertices there exists a rainbow path
then $G$ is said to be $k$-rainbow connected.
Further, if the $k$-coloring ensures that between every pair of vertices one of its
geodesic \textit{i.e.,}
one of the shortest paths is a rainbow path, then $G$ is said to be $k$-strongly rainbow connected.
The minimum number of colors
required to (strongly) rainbow connect a graph $G$ is called the
(strong) rainbow connection number denoted by ($src(G)$, respectively) $rc(G)$. 

The concept of rainbow connectivity was recently introduced by Chartrand~et~al. in \cite{chartrand2008rainbow} as a measure
of strengthening connectivity. The rainbow connection problem, apart from being
an interesting combinatorial property, also finds an application in routing messages on cellular networks~\cite{chakraborty2008hardness}.
In their original paper\cite{chartrand2008rainbow}, Chartrand~et~al. 
determined $rc(G)$ and $src(G)$, in special cases where $G$ is a complete bipartite or multipartite graph.
Rainbow connectivity from a computational point 
of view was first studied by Caro~et~al.~\cite{caro2008rainbow} who conjectured that computing the rainbow
connection number of a given graph is NP-hard. 
This conjecture was confirmed by Chakraborty~et~al.~\cite{chakraborty2008hardness}, who proved
that even deciding whether rainbow connection number of a graph equals $2$ is NP-Complete. 
They further showed that the problem of deciding whether rainbow connection 
of a graph is at most $k$ is NP-hard 
where $k$ is an even integer. The status of the $k$-rainbow connectivity problem was left open for the case when $k$ is odd.
One of our results is to resolve this problem.
\\
\\
\noindent {\bf Our Results.} 
We present the following new results in this paper:
\begin{enumerate}
\item For every fixed $k \ge 3$, deciding whether $src(G) \leq k$, 
is NP-Complete even when $G$ is bipartite. As a consequence of 
our reduction, we show that it is NP-hard to approximate the 
problem of finding the strong connectivity of a graph by a 
factor of $n^{\frac{1}{2}-\epsilon}$, where $n$ is the number of vertices in $G$.
\item For every fixed odd $k \ge 3$, deciding whether $rc(G) \leq k$ is NP-Complete.
\item We consider the following natural extension of the $2$-rainbow connectivity problem:
Given a graph $G$, 
determine the maximum number of pairs of vertices that can be rainbow connected with two colors. 
We show that the above problem is {\em fixed parameter tractable} when the number of
pairs to be rainbow connected is a parameter.
\item We extend the notion of rainbow connectivity for directed graphs and show that
for a directed graph $G$ it is NP-Complete to decide whether $rc(G) \le 2$.
\end{enumerate}
In \cite{chakraborty2008hardness},
Chakraborty~et~al. 
introduced
the problem of {\em subset rainbow connectivity}, where in addition to the
graph $G = (V, E)$ we are given a set $P$ containing pairs of vertices. The goal is to answer
whether there exists an edge coloring of $G$ with $k$ colors such that every
pair in $P$ has a rainbow path. We also use the subset rainbow connectivity problem
and analogously define the subset strong rainbow connectivity problem to prove our
hardness results.

\noindent {\bf Related Work.} 
The problem of rainbow connectivity has received considerable attention after it was introduced by Chartrand~et~al. in \cite{chartrand2008rainbow}.
Caro~et~al.~\cite{caro2008rainbow}, Krivelevich~et~al.~\cite{Krivelevich:2010:RCG:1753069.1753071},
Chandran~et~al.~\cite{chandran2010rainbow} gave lower bounds for rainbow connection 
number of graphs as a function of the number of vertices and the minimum degree of the graph. 
Upper bounds were also given by Chandran~et~al.~\cite{chandran2010rainbow} 
for special graphs like interval graphs and AT-free graphs. In~\cite{basavaraju2010rainbow}, 
Basavaraju~et~al. gave a constructive argument to show that any graph $G$ 
can be colored with $r(r+2)$ colors in polynomial time 
where $r$ is the radius of the graph. 
The threshold function for random graph to have $rc(G)=2$ was studied by Caro~et~al.~\cite{caro2008rainbow}. 
In case of strong rainbow connection number, Li~et~al.~\cite{li2010strong} and Li and Sun~\cite{li2010src}
gave upper bounds on some special graphs. Interestingly, no good upper bounds are known for the strong rainbow connection number
in the general case.
 




\section{Strong rainbow connectivity}
\label{sec:src}
In this section, we prove the hardness result for the $k$-strong rainbow connectivity problem: given a graph $G$ 
and an integer $k \ge 3$, decide whether $src(G) \leq k$. In order to show the hardness of this problem, we 
first consider an intermediate problem called the {\em $k$-subset strong rainbow connectivity} problem. The input to the $k$-subset strong rainbow connectivity problem is a graph $G = (V, E)$ along with a set of pairs $P = \{(u, v): (u,v) \subseteq V \times V\}$ and an integer $k$. Our goal is to answer whether there exists
an edge coloring of $G$ with at most $k$ colors such that every pair $(u,v) \in P$ has a geodesic rainbow path.

The overall plan is to prove that $k$-subset strong rainbow connectivity is NP-hard by showing a reduction
from the vertex coloring problem. We then establish the polynomial time equivalence of the $k$-subset strong rainbow connectivity problem
and the $k$-strong rainbow connectivity problem. 


\subsection{$k$-subset strong rainbow connectivity}
Let $G = (V, E)$ be an instance of the $k$-vertex coloring problem. The problem is to decide whether $G$ can be vertex colored 
using $k$ colors if there exists an assignment of at most $k$ colors to the vertices of $G$
such that no pair of adjacent vertices are colored using the same color. This is one of the most
well-studied problems in computer science and is known to be NP-hard for $k \ge 3$.
Given an instance $G = (V, E)$ of the $k$-vertex coloring problem, we construct
an instance $\langle G' = (V', E'), P\rangle$ of the $k$-subset strong rainbow connectivity problem.

The graph $G'$ that we construct is a star, with one leaf vertex corresponding
to every vertex $v \in V$ and an additional central vertex $a$.
The set of pairs $P$ captures the edges in $E$, that is, for every edge $(u,v) \in E$ 
we have a pair $(u,v)$ in the set $P$. The goal is to color the edges of $G'$
using at most $k$ colors such that every pair in the set $P$ has a geodesic rainbow path.
More formally, we define the parameters $\langle G' = (V', E'), P \rangle$ of
the $k$-subset strong rainbow connectivity problem below:
\begin{eqnarray*}
V' = \{a\} \cup V; \hspace{0.2in}
E' = \{(a, v): v \in V\} \\
P  = \{(u, v): (u, v) \in E\}; \hspace{0.2in}
\end{eqnarray*}
We now prove the following lemma which establishes the hardness of the $k$-subset strong rainbow connectivity problem.

\begin{lemma}
\label{lem:vc-subset-equiv}
The graph $G = (V, E)$ is vertex colorable using $k ( \ge 3)$ colors iff the graph $G' = (V', E')$ can be
edge colored using $k$ colors such that for every pair $(u,v) \in P$ there is a geodesic rainbow path
between $u$ and $v$ in $G'$.
\end{lemma}
\begin{proof}
Assume that $G$ can be vertex colored using $k$ colors; we show an assignment of colors to the edges
of the graph $G'$. Let $c$ be the color assigned to a vertex $v \in V$; we assign the color $c$ to
the edge $(a, v) \in E'$. Now consider any pair $(u,v)  \in P$. Recall that $(u,v) \in P$ because
there exists an edge $(u,v) \in E$. Since the coloring was a proper vertex coloring of $G$, the edges $(a, u)$
and $(a, v)$ in $G'$ are assigned different colors by our coloring. Thus, the path $u- a- v$ is a rainbow
path; further since that is the only path between $u$ and $v$ it is also a geodesic rainbow path.

To prove the other direction, assume that there exists an edge coloring of $G'$ using $k$ colors
such that between every pair of vertices in $P$ there is a geodesic rainbow path. It is easy to see that
if we assign the color $c$ of the edge $(a, v) \in E'$ to the vertex $v \in V$, we get a coloring
that is a proper vertex coloring for $G$.
\end{proof}
Recall the problem of subset rainbow connectivity where we are content with any rainbow
path
between every pair in $P$.
Note that our graph $G'$ constructed in the above reduction is a tree, in fact
a star and hence between every pair of vertices in $P$ there is exactly one path. Thus,
all the above arguments apply for the $k$-subset rainbow connectivity problem as well.
As a consequence we can conclude the following:

\begin{lemma}
\label{lem:hardness-src}
For every $k \ge 3$, both the problems $k$-subset strong rainbow connectivity and
$k$-subset rainbow connectivity are NP-hard even when the input graph $G$ is a star.
\end{lemma}



\subsection{$k$-strong rainbow connectivity}
In this section, we establish the hardness of deciding whether a given graph can be strongly rainbow connected using $k$ colors. 

\begin{theorem}
\label{srcresult}
For every $k \ge 3$, deciding whether a given graph $G$ can be strongly rainbow colored using $k$ 
colors is NP-hard. Further, the hardness holds even when the graph $G$ is bipartite.
\end{theorem}

\begin{proof}
We reduce the $k$-subset strong rainbow connectivity problem to the $k$-strong rainbow connectivity problem. Let $\langle G = (V, E), P\rangle$ 
be an instance of the $k$-subset strong rainbow connectivity problem. Using Lemma~\ref{lem:hardness-src},
 we know that
$k$-subset strong rainbow connectivity is NP-hard even when $G$ is a star as well
as the pairs $(v_i, v_j) \in P$ are such that both $v_i$ and $v_j$ are leaf nodes of the star.
We assume both these properties on the input $\langle G, P\rangle$ and use them crucially in our reduction.
Let us denote the central vertex of the star $G$ by $a$ and the leaf vertices by $L=\{v_1, \ldots ,v_n\}$, that is,
$V = \{a\} \cup L$. Using the graph $G$ and the pairs $P$,
we construct the new graph $G' = (V', E')$ as follows: for every leaf node $v_i \in L$, we introduce two new
vertices $u_i$ and $u_i'$. For every pair of leaf nodes $(v_i, v_j) \in (L \times L) \setminus P$, we introduce two new vertices 
$w_{i,j}$ and $w_{i,j}'$. 

\begin{eqnarray*}
V' &=& V \cup V_1 \cup V_2 \\
V_1 &=& \{u_i: i \in \{1, \ldots, n\} \} \cup \{w_{i,j}: (v_i,v_j) \in (L \times L) \setminus  P\} \\
V_2 &=& \{u'_i: i \in \{1, \ldots, n\} \} \cup \{w_{i,j}': (v_i,v_j) \in (L \times L) \setminus  P\}
\end{eqnarray*}
The edge set  $E'$ is be defined as follows:
\begin{eqnarray*}
E' &=& E \cup E_1 \cup E_2 \cup E_3\\
E_1 &=& \{ (v_i, u_i): v_i \in L, u_i \in V_1\} \cup \\& & \{(v_i, w_{i,j}), (v_j, w_{i,j}):  (v_i,v_j) \in (L \times L) \setminus  P\} \\
E_2 &=& \{ (x, x'): x \in V_1, x' \in V_2\} \\
E_3 &=& \{ (a, x'): x' \in V_2 \}
\end{eqnarray*}
We now prove that $G'$ is $k$-strong rainbow connected iff $\langle G, P \rangle$ is $k$-subset strong rainbow connected.
To prove one direction, we first note that,
for all pairs $(v_i, v_j) \in P$, there is a two length path $v_i - a - v_j$ in $G$ and this path is also present in $G'$.
Further, this path is the only two length path in $G'$ between $v_i$ and $v_j$; hence any strong rainbow
coloring of $G'$ using $k$ colors must make this path a rainbow path. This implies that if $G$ cannot be edge
colored with $k$ colors such that every pair in $P$ is strongly rainbow connected, the graph $G'$ cannot be strongly
rainbow colored using $k$ colors.

To prove the other direction, assume that there is an edge coloring $\coloring: E \rightarrow \{c_1, c_2, \ldots, c_k\}$ 
of $G$ such that all pairs
in $P$ are strongly rainbow connected. We extend this edge coloring of $G$ to an edge coloring of $G'$, denoted by $\coloring'$, such that $G'$ is strong rainbow connected:
\begin{itemize}
\item We retain the color on the edges of $G$, {\em i.e.} $\coloring'(e) = \coloring(e): e \in E$. 
\item For each edge $(v_i, u_i) \in E_1$, we set $\coloring'(v_i, u_i) = c_3$.
\item For each pair of edges  $\{(v_i, w_{i,j}), (v_j, w_{i,j})\} \in E_1$, we set $\coloring'(v_i, w_{i,j}) = c_1$,
$\coloring'(v_i, w_{i,j}) = c_2$ (Assume without loss of generality that $i < j$). 
\item The edges in $E_2$ form a complete bipartite graph between the vertices in $V_1$ and $V_2$.
To color these edges, we  pick a perfect matching $M$ of size $|V_1|$ and assign $\coloring'(e) = c_1, \forall e \in E_2 \cap M$
and $\coloring'(e) = c_2, \forall e \in E_2 \setminus M$.
\item Finally, for each edge $(a, x') \in E_3$, we set $\coloring'(a, x') = c_3$.
\end{itemize}
It is straightforward to verify that this coloring makes $G'$ strong rainbow connected. This completes the proof of NP-hardness of the $k$-strong rainbow connectivity problem.  
\par We further note that the graph $G'$ constructed above is in fact
bipartite. The vertex set $V'$ can be partitioned into two sets $A$ and $B$, where $A = \{a\} \cup V_1$
and $B = L \cup V_2$ such that there are no edges between vertices in the same partition. This establishes the fact that the $k$-strong rainbow connectivity problem is NP-hard even for the bipartite case.
\end{proof}
%
From the same construction when $k = 3$, it follows
that deciding whether a given graph $G$ can be rainbow colored using at most $3$ colors is NP-hard.
To see this, note that between any pair of vertices $(v_i, v_j) \in P$, a path in $G'$ that is not
contained in $G$ is of length at least $4$ and the shortest path between $v_i$ and $v_j$ is in $G$. Further, we always color the edges $E' \setminus E$ using
$3$ colors; hence none of these paths can be rainbow path. Thus, we conclude the following corollary.

\begin{corollary}
Deciding whether $rc(G) \le 3$ is NP-hard even when the graph $G$ is bipartite.
\end{corollary}

As a consequence of the reduction from the $k$-subset strong rainbow connectivity to the $k$-strong rainbow connectivity, we have the following result:

\begin{theorem}
\label{thm:inapprox}
There is no polynomial time algorithm that approximates strong rainbow connection number
of a graph $G = (V, E)$ within a factor of $n ^{\frac{1}{2} - \epsilon}$, unless $NP=ZPP$.
Here $n$ denotes the number of vertices of $G$.
\end{theorem}
\begin{proof}
In order to prove the NP-hardness of the $k$-strong
rainbow connectivity problem, we started with an instance $G = (V, E)$ of the
$k$-vertex coloring problem, and obtained an instance $\langle G' = (V', E'), P\rangle$ of the
$k$-subset strong rainbow connectivity problem, where $P$ denotes the set of pairs. From $\langle G', P \rangle$ we obtained an
instance $G'' = (V'', E'')$ of the $k$-strong rainbow connectivity problem. 
Further, note that, $G$ is vertex colorable using at most $k$ colors if and only if $src(G'') \leq k$.
Let $n$ denote the number of vertices of $G$, then the graph $G'$ has $n+1$ vertices.
We now bound the number of vertices in the graph $G''$. 
Recall that during the construction of $G''$ we added two new vertices for every vertex in $G'$ 
and two new vertices for every pair in $P$. Thus the total number of vertices (denoted by $N$) in $G''$ can be
bound from above as:
\begin{eqnarray*}
N &\leq& 2 {n \choose 2} + 3n + 3 \\
&\leq& n^2 + 2n + 3 \\
&\leq& 2n^2
\end{eqnarray*}
Thus, if there is a ${N}^{\frac{1}{2} - \epsilon}$ approximation algorithm for computing strong
rainbow connectivity of a graph, then there is a $(\sqrt{2}n)^{1 - \epsilon}$ approximation algorithm
for the vertex coloring problem. 
But we know from \cite{feige1996zero}, that it is hard to approximate chromatic number within a factor of $\Omega(n^{1-\epsilon})$ unless $NP=ZPP$ and hence the result.  
\end{proof}

\section {Rainbow connectivity}
\label{rc}

In this section we investigate the complexity of deciding whether the rainbow connection
number of a graph $G$ is at most $k$. 
We prove the NP-hardness of the $k$-rainbow connectivity problem i.e., deciding whether $rc(G) \leq k$, when $k$ is odd.
We recall from Lemma~\ref{lem:hardness-src} that the $k$-subset rainbow connectivity problem 
is NP-hard.  In the following theorem, we give a reduction of the $k$-subset rainbow connectivity problem to the $k$-rainbow connectivity problem. 
\begin{theorem}
\label{rcresult}
For every odd integer $k \geq 3$, deciding whether $rc(G) \leq k$ is NP-Complete.
\end{theorem}
\begin{proof}
Let $\langle G = (V, E), P \rangle$ be an instance of the $k$-subset rainbow connectivity
problem. Since $k$ is assumed to be odd, let $k = 2m + 1$ where $m \in \mathbb{N}$.
Let us denote the vertices of $G$ as $V=\{v_1, \ldots ,v_n\}$. 
Given the graph $G$ and a set of pairs of vertices $P$, we construct an instance $G' = (V', E')$
of the $k$-rainbow connectivity problem as follows:
For each vertex $v_i \in V$, we add $2m$ new vertices denoted by $u_{i, j}$ where $j \in \{1, \ldots, 2m\}$.
Further, we add the following two paths: $v_i- u_{i,1}- u_{i,2} \cdots -u_{i,m}$ and 
$v_i- u_{i,m+1}- u_{i,m+2} \cdots -u_{i,2m}$. We also add edges $(u_{i,m}, u_{i,2m})$ and $(u_{i,1}, u_{i,m+1})$
(if $m=1$, we only add one edge). For every pair of vertices $(v_i,v_j) \in (V \times V) \setminus P$: 
we add the edges $(u_{i,m}, u_{j,2m})$ and $(u_{i,2m}, u_{j,m})$.
We add two more new vertices $x,y$ and for every $v_i \in V$
we add the following edges:  $(x,u_{i,m})$, $(x,u_{i,2m})$, $(y,u_{i,m})$ and $(y,u_{i,2m})$.
Figure~\ref{fig:long-paths} shows a subgraph of the graph $G'$. The construction shows extra vertices added corresponding to $v_i$
and $v_j$ such that the pair $(v_i, v_j) \in (V \times V) \setminus P$.
More formally, the vertex set $V'$ can be defined as:
\begin{eqnarray*}
V' &=& V \cup V_{1,m}\cup V_{m+1,2m} \cup V_{x,y} \\
V_{1,m} &=& \{u_{i,j}: v_i \in V, j \in \{1, \ldots, m\}\} \\
V_{m+1, 2m} &=& \{u_{i,j} : v_i \in V, j \in \{m+1, \ldots, 2m\}\} \\
V_{x,y} &=& \{x, y\}
\end{eqnarray*}
The edge set $E'$ can be defined as: 
\begin{eqnarray*}
E'&=& E \cup E_1 \cup E_2 \cup E_{x,y} \\
E_1 &=& \{(u_{i,j},u_{i,j+1})\ :\ v_i \in V,\ j \in \{1, \ldots ,m\},\ j \pmod m \neq 0\} \cup \\
&\ & \{(u_{i,1},u_{i,m+1}),(u_{i,m},u_{i,2m}): v_i \in V\} \\
E_2 &=& \{(u_{i,m}, u_{j,2m}), (u_{i,2m}, u_{j,m})\ :\ (v_i,v_j) \in (V \times V) \setminus P \} \cup \\
&=& \{(v_i,u_{i,1}),(v_i,u_{i,m+1}): v_i \in V\}\\
E_{x,y} &=& \{(x,u_{i,m}),(x,u_{i,2m}),(y,u_{i,m}),(y,u_{i,2m})\ :\ i \in \{1, \ldots ,n\} \}
\end{eqnarray*}

\begin{figure}
\begin{center}
\psfrag{G}{$G$}
\psfrag{x}{$x$}
\psfrag{y}{$y$}
\psfrag{vi}{$v_i$}
\psfrag{vj}{$v_j$}
\psfrag{ui1}{$u_{i,1}$}
\psfrag{uj1}{$u_{j,1}$}
\psfrag{ui2}{$u_{i,2}$}
\psfrag{uim}{$u_{i,m}$}
\psfrag{uim1}{$u_{i,m+1}$}
\psfrag{ujm1}{$u_{j,m+1}$}
\psfrag{ujm2}{$u_{j,m+2}$}
\psfrag{uj2m}{$u_{j,2m}$}
\includegraphics[scale=0.5]{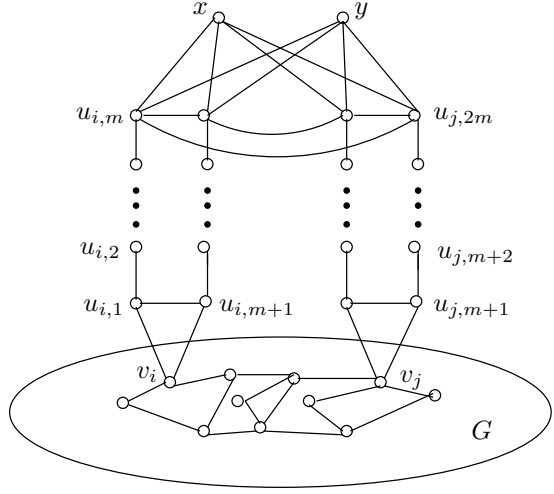}
\end{center}
\caption{A subgraph of the graph $G'$. The construction shows extra vertices added corresponding to vertices $v_i$ and $v_j$ 
belonging to $G$. The pair $(v_i, v_j) \in (V \times V) \setminus P$.}
\label{fig:long-paths}
\end{figure}

We claim that $G$ can be edge colored using $k$ colors such that every pair belonging
to $P$ is rainbow connected if and only if $rc(G') \le k$.
Assume that $G$ can be edge colored using $k$ colors such that all pairs 
in $P$ are rainbow connected. Let $\chi:E \rightarrow \{c_1, \ldots, c_{2m+1}\}$ 
be such a coloring. We obtain a coloring $\chi':E' \rightarrow \{c_1, \ldots, c_{2m+1}\}$ as follows:
\begin{itemize}
\item  For every $v_i \in V$: \\
 $\chi'(v_i,u_{i,1})=c_1$; $\chi'(v_i,u_{i,m+1})=c_{m+1};$ \\
 $\chi'(u_{i,j},u_{i,j+1})=c_{j+1}$ where $j \in \{1, \ldots, 2m-1\}$ and $j \pmod m \neq 0;$  \\
 $\chi'(x,u_{i,m})=c_{m+1};\ \chi'(x,u_{i,2m})=c_{2m+1};$ \\
 $\chi'(y,u_{i,m})=c_{2m+1}; \ \chi'(y,u_{i,2m})=c_1$. \\
If $m \neq 1$, $ \chi'(u_{i, 1}, u_{i, m+1}) = c_{m+1}; \ \chi'(u_{i, m}, u_{i, 2m}) = c_1$ \\
else $ \chi'(u_{i, 1}, u_{i, m+1}) = c_1$.
\item  For every $(v_i,v_j) \in (V \times V) \setminus P$: $\chi'(u_{i,m},u_{j,2m})=c_{2m+1}$ and $\chi'(u_{i,2m},u_{i,m})=c_{2m+1}$.
\item  For every edge $(v_i,v_j) \in E$: $\chi'(v_i,v_j)=\chi(v_i,v_j)$.
\end{itemize}
We claim that if $\chi$ makes the graph $G$ $k$-subset rainbow connected then $\chi'$ makes the graph $G'$ $k$-rainbow connected. To prove that, consider the following cases:
\begin{enumerate}
\item [-] Every vertex $v_i \in V$ has a rainbow path to every vertex $u \in V'\setminus V$. 
Indeed, if $u \in V_{x,y}$ then the path $v_i-u_{i,1}-\cdots- u_{i,m}-u$ 
is a rainbow path connecting $v_i$ and $u$. If $u = u_{j,s}, s \in \{1, \ldots, m\}, v_j \in V$ 
then: if $i=j$, then the path $v_i-u_{i,1}- \cdots- u_{i,s}$ 
is the rainbow path otherwise the path $v_i-u_{i,m+1}-\cdots u_{i,2m}-y-u_{j,m}-u_{j,m-1}-\cdots- u_{j,s}$ 
is the rainbow path connecting $v_i$ and $u$. 
Now if $u = v_{j,s}, s\in \{m+1, \ldots, 2m\}$ then: 
if $i=j$, then the path $v_i-v_{m+1}\ldots v_{s}$ otherwise 
the path $v_i-u_{i,1}-\cdots- u_{i,m}-x-u_{j,2m}-u_{j,2m-1}-\cdots-u_{j,s}$ is the rainbow path between $v_i$ and $u$.

\item [-] Every vertex $u \in V_{x,y}$ has a rainbow 
path to every vertex $u' \in V_{1,m} \cup V_{m+1,2m}$. 
Indeed, if $u'=u_{i,s},s\in \{1, \ldots, m\}$ the rainbow path 
is $u_{i,s}-\cdots-u_{i,m}-u$ and if 
$u' = u_{i,s}, s \in \{m+1,\ldots, 2m\}$ 
the rainbow path is $u_{i,s}-\cdots-u_{i,2m}-u$. 
Also $x,y$ are connected by a rainbow path $x-u_{i,m}-y$, for some $v_i \in V$.

\item [-] Every vertex $u_{i,s} \in V_{1,m}$ is rainbow connected to every vertex $u_{j,r} \in V_{m+1,2m}$: $u_{i,s}-\cdots-u_{i,m}-y-u_{j,2m}-u_{j,2m-1}-\cdots-u_{j,r}$ is a rainbow path. Every vertex $u_{i,s} \in V_{1,m}$ is rainbow connected to $u_{j,r} \in V_{1,m}$ (without loss of generality let $s \geq r$): if $i=j$ then the rainbow path is $u_{j,r}-u_{j,r+1}-\cdots-u_{i,s}$ otherwise $u_{j,r}-u_{j,r-1}- \ldots u_{j,1}-u_{j,m+1}-\cdots-u_{j,2m}-y-u_{i,m}-u_{i,m-1}-\cdots-u_{i,s}$ is the rainbow path. Every vertex $u_{i,s} \in V_{m,2m}$ is rainbow connected to every vertex $u_{j,r} \in V_{m,2m}$ (without loss of generality assume $s\geq r$): if $i = j$ then the rainbow path is $u_{j,r}-u_{j,r+1}-\cdots- u_{j,s}$ otherwise the path $u_{i,s}-u_{i,s+1}-\cdots- u_{i,2m}-y-u_{j,m}-u_{j,m-1}-\cdots-u_{j,1}-u_{j,m+1}-\cdots-u_{j,r}$ is a rainbow path.

\item [-] Every pair of vertices $v_i,v_j \notin P$ is rainbow connected: $v_i-u_{i,m+1}-\cdots-u_{i,2m}-u_{j,m}-u_{j,m-1}\ldots u_{j,1}-v_j$ is a rainbow connected path.

\item [-] Since every pair $v_i,v_j\in P$ is rainbow connected in $G$ which is an induced subgraph of $G'$, those pairs are rainbow connected in $G'$.

\end{enumerate}
From the above cases, $G'$ is $k$-rainbow connected if $G$ is $k$-subset rainbow connected.
\par To prove the other direction, assume that $rc(G') \leq k$. 
Let $\chi:E' \rightarrow \{c_1, \ldots ,c_k\}$ be an edge coloring 
of $G'$ such that $\chi$ makes $G'$ rainbow connected. We will 
translate this edge coloring of $G'$ to an edge coloring of $G$ as follows: 
color the edge $(v_i,v_j)$ in $G$ with the color $\chi(v_i,v_j)$. 
We claim that all pairs in $P$ are rainbow connected in $G$. 
This follows from the observation that for a pair $(v_i,v_j) \in P$, 
any path between $v_i$ and $v_j$ which is of length at most $2m+1$ in $G'$ has 
to be completely contained in $G$.
Since $\chi$ makes $G'$ rainbow connected, 
the rainbow path between $v_i$ and $v_j$ in $G'$ has to lie completely 
inside $G$ itself. Correspondingly, there is a rainbow path between $v_i$ and $v_j$ in $G$. 
Hence, all pairs in $P$ are rainbow connected in $G$. This proves that $k$-rainbow connectivity problem is NP-hard.  
\par It is clear that given an edge $k$-coloring, for $k\in \mathbb{N}$, we can check in polynomial time, that the edge coloring rainbow connects every pair of vertices. Hence the problem of deciding if $rc(G)\leq k$ is in NP. The result follows.
\end{proof}

\noindent Unlike the case of strong rainbow connectivity, the reduction presented above does not give any insight into the inapproximability of the problem of finding the rainbow connection number of a graph. The reason being that the reduction stated in the proof of Theorem~\ref{srcresult} yields an instance of $k$-strong rainbow connectivity problem which is independent of $k$ \textit{i.e.,} the structure of the graph does not change with $k$. On the contrary, the size of the instance of $k$-rainbow connectivity problem obtained from the reduction in Theorem~\ref{rcresult} crucially depends on $k$. 

\subsection{Parameterized complexity}
In this section, we study the parameterized complexity of a maximization version of the rainbow connectivity problem. Before that, we describe the necessary preliminaries. A problem is said to be fixed parameter tractable (FPT) with respect to a parameter $k$\footnote{A parameter is a natural number associated to a problem instance. For example, a parameter could be the number of vertices of a graph instance in a vertex cover problem or the number of processors in a scheduling problem.}, if given an instance $x$ of size $|x|$ there exists an algorithm with running time $f(k) \times |x|^{O(1)}$ where $f$ is a function of $k$ which is independent of $|x|$. One way of showing that a problem is fixed parameter tractable is to exhibit polynomial time reductions to obtain a {\em kernel}  which is basically an equivalent instance whose size is purely a function of the parameter $k$. If the size of the kernel is a linear function in $k$ then the kernel is termed as a {\em linear kernel}. For formal definitions, we refer the reader to~\cite{rod1999parameterized,flum2006parameterized}.
\par We are interested in the following problem: Given a graph $G = (V, E)$, color the edges of $G$ using $2$ colors such
that maximum number of pairs are rainbow connected. Since deciding whether $rc(G) \le 2$ is NP-Complete
\cite{chakraborty2008hardness}, it follows that the above maximization problem is NP-hard.
Any edge coloring of a graph $G = (V, E)$ with $2$ colors, trivially satisfies $|E|$ pairs. Hence,
we are interested in deciding whether $G$ can be $2$-colored such that at least $|E| + k$ pairs
of vertices are rainbow connected, where $k$ is a parameter. We show that the problem is {\em fixed
parameter tractable} with respect to $k$. 
We first state a useful lemma. Let us call a {\em non-edge} in $G$ as an anti-edge; formally we call $e = (u, v)$ an
anti-edge of a graph  $G = (V, E)$ if $e \notin E$.
\begin{lemma}
Let $G = (V, E)$ be a connected graph with at least $k$ anti-edges and a clique of size $\geq k$.
The edges of $G$ can be colored using $2$ colors such that at least $|E| + k$ pair of vertices
are rainbow-connected.
\label{lem:cliq}
\end{lemma}
\begin{proof}
Let $M$ be a maximal clique of size $\geq k$. Let $L_i$
denote the set of vertices in $G$ which are at a distance $i$ from at least
one vertex of $M$. Since there are at least $k$ anti-edges and the graph is connected,
$L_1$ is non-empty. For $L_2$ we have two cases -- either $L_2$ is empty or $L_2$ is non-empty.
First consider the case when $L_2$ is non-empty.
We now show a coloring of the edges to ensure that $|E| +k$ pairs are rainbow connected.
Let us color the edges of the clique $M$ by red and the edges from $M$ to $L_1$ by blue.
Finally, color the edges from $L_1$ to $L_2$ as red. Let $u_2 \in L_2$ be adjacent to some
$u_1 \in L_1$. If $r$ vertices of $M$ are not adjacent to $u_1$, then at least $k-r$ vertices of $M$ are
adjacent to $u_1$ and hence have a path of length $2$ to $u_2$. The $r$ vertices of $M$ paired with $u_1$
and the rest of the vertices of $M$ paired with $u_2$ form the required $k$ anti-edges that
are rainbow connected.

It remains to deal with the case when $L_2$ is empty. If there are $k$ anti-edges from $M$ to
$L_1$ we can color the edges of $M$ by red and edges from $L_1$ to $M$ by blue
and we are done, else we have less than $k$ anti-edges between $L_1$ and
$M$. In this case, it is easy to observe that
there should be at least $2$ vertices in $L_1$ and every vertex in $L_1$
has at least $2$ neighbours in $M$. Consider any pair of vertices $u$ and $v$ in $L_1$.
Let $N_{M}(u)$ denotes the neighbourhood of $u$ in $M$, then we note that $N_{M}(u) \cap N_{M}(v) \neq \phi$.
Otherwise we already have $k$ anti-edges from $L_1$ to $M$. Our goal is now to rainbow connect as many pair
of anti-edges $(u, v)$ in $L_1$. We do this greedily as follows:
\begin{itemize}
\item Let $S$ denote the set of anti-edges in $L_1$, i.e.,\\ $S = \{(u, v): u \in L_1, v \in L_1 \text{ and } (u,v) \text{ is an anti-edge}\}$
\item All vertices $w \in M$ are unmarked to begin with.
\item while $ (S \neq \phi)$ do:
\begin{itemize}
\item Let $e = (u,v)$ be an anti-edge in $S$ 
such that there exists an unmarked $w \in M$ and $w \in (N_{M}(u) \cap N_{M}(v))$.
\item If no such $e$ exists, break.
\item Else, color $(u,w)$ as red and $(v, w)$ as blue. \\
$S = S \setminus \{e\}$.
\item Mark the vertex $w$.
\end{itemize}
\item end while.
\end{itemize}
Using this procedure, some anti-edges in $L_1$ are rainbow connected. We color the edges of
$M$ using red and the uncolored edges from $M$ to $L_1$ using blue.

We can assume that for every vertex $u$ in $L_1$, there is a blue edge $(u,w)$ such that $w \in M$. If not, consider a vertex $u \in L_1$ such that all the edges from $u$ to $M$ are colored red. Let there be $r$ such red edges incident on $u$. This implies that $r$ anti-edge pairs belonging to $L_1$ got rainbow connected in the above while loop. Further, note that all the vertices in $N_{M}(u)$ got marked in the while loop and hence at least $k-r$ vertices in $M$ are non-neighbors of $u$. We will rainbow connect at least $k-r$ anti-edge pairs by the following recoloring. Let $w \in N_{M}(u)$. Recolor all edges of the form $(w,w')$ where $w' \in M \setminus N_{M}(u)$ as blue. This will rainbow connect all the $(u,w')$ anti-edge pairs, thus rainbow connecting a total of $k$ anti-edge pairs.      

Now we are in the case where for every vertex $u \in L_1$, there is a blue edge incident from $u$ to $M$. Thus
all anti-edges from $L_1$ to $M$ are rainbow connected. We now consider the two cases in which we have quit the while loop. Suppose that we quit the while loop
because, the set $S$ was empty, then we are done. Because all the anti-edges in $L_1$ are rainbow connected,
further, all anti-edges from $L_1$ to $M$ are also rainbow connected. As there are at least $k$ anti-edges in $G$,
we are done. Finally, suppose the set $S$ is not empty when we quit the while loop. Assume that
we rainbow connect $r$ anti-edges in $S$ by the greedy procedure. This implies that
there are at least $k-r$ unmarked vertices in $M$. Consider an anti-edge $(u,v) \in S$ that our greedy procedure
could not rainbow connect. Then for every unmarked vertex $w \in M$, it has to be the case that either $u$ or
$v$ is non-adjacent to $w$. Thus we have at least $k-r$ anti-edges which are from $L_1$ to $M$. Thus, in this 
case also we have rainbow-connected $k$ anti-edge pairs.

This completes the proof of the lemma.

\end{proof}

Using the above lemma~\ref{lem:cliq} we now show that the problem is fixed parameter tractable. 
\begin{theorem}
Given a graph $G = (V, E)$, decide whether the edges of $G$ can be colored using $2$ colors
such that at least $|E| + k$ pair of vertices are rainbow connected.
The above problem has a kernel with at most $4k$ vertices and hence is fixed parameter tractable.
\label{stan:para}
\end{theorem}
\begin{proof}
Let $v$ be any arbitrary vertex and let $O_v$ be the set of vertices which are not adjacent to $v$.
We claim that there is a coloring which rainbow connects at least
$|O_v|$ pair of non-adjacent vertices. Consider a breadth first search (bfs)
 tree rooted at $v$. Denote the set of vertices in each level of 
the bfs tree by $L_i$, $i \geq 1$. Then, $L_1=\{v\}$, $L_2=N(v)$ and $O_v= \cup_{i > 2} L_i$.
We now color the edges from $L_{i-1}$ to $L_{i}$ by red
if $i$ is odd and by blue if $i$ is even.
For $i > 2$, every vertex of $L_i$ is rainbow connected to some vertex of $L_{i-2}$.
Thus we have $|O_v|$ pairs of non-adjacent vertices rainbow connected by this coloring.
Hence if $|O_v| \ge k$ for
any vertex $v \in V$, we have a trivial {\em yes} instance at hand.
Otherwise, $|O_v| < k$, for all $v \in V$.

Recall that our goal is to color the graph using $2$ colors such that at least $|E| + k$
pair of vertices are rainbow connected. If $G$ has less than $k$ anti-edges, clearly
this is not possible and we have a {\em no} instance. Assume that this is not the case.
Now consider a vertex $v$ and let $N(v)$ denote the neighbors of $v$ in $G$. Let $H$
denote the complement of the graph induced by the neighbourhood of $v$, ie the complement of $G[N(v)]$\footnote{$G[H]$ denotes the induced subgraph of $G$ on vertices of $H$}
Further, let $C_1, C_2,\ldots, C_r$ denote the components of $H$. If there are more than
$k$ isolated vertices in $H$, we have a clique of size $\ge k$ in $G$. Further, since there
are at least $k$ anti-edges, using lemma~\ref{lem:cliq}, we have a coloring which rainbow connects
at least $|E| + k$ pairs of $G$. Thus we have a {\em yes} instance.

It remains to deal with the
case when the number of isolated vertices in $H$ is less than $k$. Let $C_i$ be some non-trivial component of $H$, that
is $C_i$ contains at least two vertices. (If no non-trivial component exists,
we are already done, since we can bound the number of vertices of $G$ from above
by $2k$). We now show a coloring of edges of $G$ such that
at least $|C_i| - 1$ non-adjacent vertices are rainbow connected. For this, consider
a spanning tree of $C_i$ and color the vertices of the spanning tree level by level using alternate colors.
That is, color the root as red, the vertices at the next level in the spanning tree as blue and so on.
We map the colors on the vertices of $C_i$ back to the edges of $G$ as follows. If a vertex $u_1 \in C_i$
got the color red, we color the edge $(v, u_1) \in G$ as red. Thus for every edge $(u_1, u_2)$ in $C_i$ that
got distinct colors on its end points, we ensure that one pair got rainbow connected via the path
$u_1, v, u_2$. Further, since $(u_1, u_2)$ is an edge in $H$, it is an anti-edge in $G$. Thus for every non-trivial
component $C_i$ we can rainbow connect $|C_i| - 1$ anti-edges of $G$. Counting this across all
the components we have $\sum_{i=1}^r |C_i| - r$ pairs of anti-edges in $G$ rainbow connected.
If $\sum_{i=1}^r |C_i| - r \ge k$ we have a {\em yes} instance, otherwise we have:

\begin{equation}
\label{eq1}
\Sigma_{i=1}^{r} |C_i| - r < k.
\end{equation}
Let the number of non-trivial components of $H$ be $s$. Each of these non-trivial
components have at least $2$ vertices. Hence we have the following:
\begin{equation}
\label{eq2}
\Sigma_{i=1}^{r} |C_i| \geq 2*s + (r-s) = r + s
\end{equation}
Since the number of isolated vertices in $H$ is strictly less than $k$,
we have $r < s+k$. Further, from equations (\ref{eq1}) and (\ref{eq2}) we get $s < k$. Combining these
we have $r < 2k$.
Thus we can bound the number of vertices in $H$ as:
\begin{equation}
|H| = \Sigma_{i=1}^{r} |C_i|  < r + k < 3k
\end{equation}
Therefore we have:
\begin{equation}
|G| = |H| + 1 + |O_v| < 3k + 1 + k \implies |G| \leq 4k.
\end{equation}
Hence, we have a $4k$ vertex kernel.
\end{proof}

\section{Rainbow connectivity on directed graphs}
\label{rcd}
In this section, we consider the rainbow connectivity problem
for directed graphs. All the directed graphs considered in this section are assumed to be connected \textit{i.e.}, between any two vertices $u,v$ in the directed graph there is either a directed path from $u$ to $v$ or from $v$ to $u$.  
Consider an edge-coloring of  
a directed graph $G=(V,E)$. We say that there exists a rainbow path
between a pair of vertices $(u,v)$ if there exists a
directed path from $u$ to $v$ or from $v$ to $u$ with distinct edge colors.
 An edge coloring of the
edges in a directed graph is said to make the graph rainbow connected if
between every pair of vertices there is a rainbow path. Analogous to the undirected version,
the minimum colors needed to rainbow color a
directed graph $G$ is called the rainbow connection number
of the directed graph. The rainbow connection number of a directed graph is at least the rainbow connection number of the underlying undirected graph; however, there are examples where the directed graph requires many more colors than the underlying undirected graph. Consider the directed graph $G=(V,E)$ with $V=\{v_1, \ldots ,v_n\}$ and $E=\{(v_i,v_{i+1}): i=1, \ldots ,n-1 \} \cup \{(v_1,v_n)\}$. The rainbow connection number of $G$ is $n-2$ while the rainbow connection number of its underlying undirected graph, which is a cycle, is $\lceil \frac{n}{2} \rceil$.
\par We study the computational complexity of the problem of computing rainbow connection number for a directed graph. We prove that the problem of deciding
whether the rainbow connection of a simple directed graph is at most
$2$ is NP-hard. As in the case of undirected graphs, 
we define the problem of subset rainbow connectivity on directed graphs. Given
a directed graph $G = (V, E)$ and a set of pairs $P \subseteq V \times V$ decide whether
the edges of $G$ can be colored using $2$ colors such that every pair in $P$ 
is rainbow connected (in the directed sense). Throughout this section we will use the
term rainbow connected to mean that it is rainbow connected in the directed sense.
Our plan, as in the previous cases, is to show that the $2$-subset rainbow rainbow connectivity is NP-hard by showing a reduction from the 3SAT problem. We then establish the polynomial time equivalence 
of $2$-subset rainbow connectivity and $2$-rainbow connectivity for a directed graph $G$.

Let $\mathcal{I}$ be an instance of the 3SAT problem with $X = \{x_1, \ldots, x_n\}$ as the
set of variables and $C_1, \ldots, C_m$ being the clauses. We construct from $\mathcal{I}$
a directed graph $G = (V, E)$ and a set of pairs $P \subseteq V \times V$ which is an instance
of the 2-subset rainbow connectivity problem. For readability sake, we reuse the symbols $C_i,x_i$ to represent the vertices.
\begin{eqnarray*}
V &=& \{C_i: i\in \{1, \ldots, m\}\} \cup X \cup \bar{X}  \cup \{T,R,B\} \\
\bar{X} &=& \{\bar{x_i}: x_i \in X\}
\end{eqnarray*}
The edge set $E$ is defined as below. We say that $x_i \in C_j$ to imply that the clause $C_j$ contains
the positive occurrence of the variable $x_i$. If $x_i$ appears negated in the clause $C_j$
we denote it as $\bar{x_i} \in C_j$.
\begin{eqnarray*}
E &=& \{(R,T),(T,B)\}  \cup  \\
  & \ &  \{(x_i,T),(T,\bar{x_i}),(x_i,\bar{x_i}) : x_i\in X\} \cup \\
  & \ & \{(C_j,x_i) : x_i \in C_j\} \cup \\
  & \ & \{(\bar{x_i}, C_j): \bar{x_i} \in C_j\}
\end{eqnarray*}
The set of pairs $P$ is defined as follows:
\begin{eqnarray*}
P &=& \{(C_i,T) : i \in \{1, \ldots, m\} \} \cup  \\
  & \ & \{(x_i,C_j),(\bar{x_i},C_j) : x_i \in C_j\} \cup \\
  & \ & \{(x_i,C_j),(\bar{x_i},C_j) : \bar{x_i} \in C_j\} \cup \\
  & \ &  \{(R,B)\} \cup \{(R, \bar{x_i}),(B,x_i) : x_i\in X\} 
\end{eqnarray*}

We now state the following lemma which establishes the correctness of our reduction.
\begin{lemma}
\label{lem:dir1}
There exists a satisfying assignment for $\mathcal{I}$ if and only if there is an 
edge coloring of $G = (V, E)$ with $2$ colors such that all the pairs in $P$ are
rainbow connected.
\end{lemma}
\begin{proof}
Let us assume that $\mathcal{I}$ has a satisfying assignment. Using the assignment,
we will show a coloring of the edges $\chi$ of $G$ using $2$ colors $red$ and $blue$ 
such that all pairs in $P$ are rainbow connected.
\begin{itemize}
\item $\chi(R, T) = red$; $\chi(T, B) = blue$.
\item $\chi(x_i, T) = red$ for $x_i \in X$ and $\chi(T, \bar{x_i}) = blue$ for $x_i \in X$.
\item If $x_i$ is set to $true$ in the satisfying assignment, we set $\chi(C_j, x_i) = blue$ for $x_i \in C_j$;
$\chi(\bar{x_i}, C_j) = blue$ for $\bar{x_i} \in C_j$ and $\chi(x_i, \bar{x_i}) = red$.
If $x_i$ is set to $false$ in the satisfying assignment, set $\chi(C_j, x_i) = red$ for $x_i \in C_j$;
$\chi(\bar{x_i}, C_j) = red$ for $\bar{x_i} \in C_j$ and $\chi(x_i, \bar{x_i}) = blue$
\end{itemize}
To see that the above coloring rainbow connects all the pairs in $P$, we
first note that this is trivially true for all pairs of $P$ which involve $R$ or $B$. Since the coloring is obtained from a satisfying assignment of $\mathcal{I}$, every
clause has some literal which is set to true. Let $x_i \in C_j$ be set to true, then
the path $C_j, x_i, T$ is a directed rainbow path from $C_j$ to $T$. Else if
$\bar{x_i} \in C_j$ is set to true then the path $T, \bar{x_i}, C_j$ is a
rainbow path from $T$ to $C_j$. Finally, it is easy to observe that if $x_i \in C_j$
or $\bar{x_i} \in C_j$, then
we have a rainbow path connecting $x_i$ to $C_j$ and $\bar{x_i}$ to $C_j$.
Thus the coloring ensures that $G$ is $2$-subset rainbow connected.

To prove the other direction, assume that $G$ can be edge-colored using $2$ colors
such that all the pairs in $P$ have a rainbow path. We show that $\mathcal{I}$ has a satisfying assignment.
Assume, without loss of generality, that the edge $(R, T)$ is colored $red$. The color on the edge $(T,B)$ has
to be different from the color on the edge $(R, T)$ since $(R, B) \in P$ and the path $R-T-B$ is the
only directed path in $G$. Hence color of the edge $(T, B)$ is $blue$. This forces all the edges $\{(T, \bar{x_i}): x_i \in X\}$
to be colored blue and all the edges $\{(x_i, T): x_i \in X\}$ to be colored $red$. (This is because
we have pairs $(R, \bar{x_i})$ and $(B, x_i)$ in $P$.) We also observe that for any variable $x_i \in X$
edges of the form $(C_j, x_i)$ and edges of the form $(\bar{x_i}, C_j)$ have to get the same color
by the construction of our pairs. We now construct a truth assignment in the following way:
\begin{enumerate}
\item For any $x_i$, if any edge $(C_j, x_i)$ incident on it is colored $blue$, assign $x_i$ as $true$ else
assign $x_i$ as $false$.
\item If no edge of the form $(C_j, x_i)$ is incident on $x_i$, assign $x_i$ false.
\end{enumerate}
It is easy to verify that this assignment is a satisfying assignment for $\mathcal{I}$
since the graph $G$ is $2$-subset rainbow connected.
\end{proof}

\noindent We now prove the equivalence of the following two problems. 
\begin{lemma}
\label{lem:dir2}
The following two problems are polynomial time equivalent:\\
(1) Given a directed graph $G = (V, E)$ decide whether $G$ is $2$-rainbow connected. \\
(2) Given a directed graph $G = (V, E)$, and a set of pairs $P \subseteq V \times V$,
decide whether $\langle G, P \rangle$ is $2$-subset rainbow connected.
\end{lemma}
\begin{proof}
It suffices to prove that problem (2) reduces to problem (1).
Given $\langle G = (V, E), P\rangle$ we construct an instance $G' = (V', E')$ as follows:
\begin{eqnarray*}
V' &=& V \cup V_1 \cup \{v_{ex}\} \\
V_1 &=& \{ w_{i,j} : (v_i, v_j)  \in (V \times V) \setminus P, v_i \neq v_j \}
\end{eqnarray*}
The edge set $E'$ is defined as:
\begin{eqnarray*}
E' &=& E \cup \{(v_i, w_{i,j}), (w_{i, j}, v_j) : (v_i, v_j)  \in (V \times V) \setminus P, v_i \neq v_j \} \cup \\
   & \ & \{ (v, v_{ex}), (v_{ex}, x) : v \in V, x \in V_1 \} \cup E_1
\end{eqnarray*}
The set of edges in $E_1$ are amongst the vertices in $V_1$ such that
the induced subgraph $T = (V_1, E_1)$ is a tournament. 
\par Assume that $G$ has an edge coloring $\chi$ using two colors, say $red$
and $blue$ such that
every pair of vertices in $P$ is rainbow connected. We give a coloring
$\chi'$ the edges
of $G'$ as follows:
\begin{itemize}
\item Set $\{\chi'(v, v_{ex}) = red : v \in V\}$ 
and set $\{\chi'(v_{ex}, x) = blue : x \in V_1\}$.
\item For every pair $(v_i, v_j)  \in (V \times V) \setminus P$, we set
$\chi'(v_i, w_{i,j}) = red$ and $\chi'(w_{i,j}, v_j) = blue$.
\item Color the edges of the graph induced by $V_1$ arbitrarily.
\item Set $\{\chi'(v_i, v_j) = \chi(v_i, v_j) : v_i \in V, v_j \in V\}$.
\end{itemize}
It is easy to verify that the above coloring makes $G'$ rainbow connected.

In the other direction, we note that no pair of vertices in $P$ has a 
directed $2$ length path in $G'$ which is not contained entirely in $G$.
Hence if $G'$ has an edge coloring using $2$ colors such that every pair has
a rainbow path, then the coloring of the induced subgraph $G$ of $G'$ rainbow
connects every pair of vertices in $P$.
This completes the proof of the lemma.
\end{proof}
Using lemma~\ref{lem:dir1} and lemma~\ref{lem:dir2} we can conclude the following theorem.
\begin{theorem}
Given a directed graph $G = (V, E)$, it is NP-hard to decide whether $G$ can be
colored using two colors such that between every pair of vertices there is a rainbow path.
\end{theorem}

\section{Conclusion}
\label{concl}
In this paper, we present several hardness results related to the rainbow connectivity
problem. The hardness results for the strong rainbow connectivity and rainbow connectivity
problem are due to a series of reductions starting from the vertex coloring problem. 
Our study on parameterized version of the rainbow
connectivity problem shows a linear kernel when we want to maximize the number
of pairs which are rainbow connected using two colors. 
We initiate the study of rainbow connectivity in directed graphs. Further, we show that the problem of deciding whether a directed graph can be rainbow connected
using at most $2$ colors is NP-hard. 

\section{Acknowledgements}   
The first and the second author would like to thank Deepak Rajendraprasad and Dr.~L.~Sunil Chandran for the useful discussions on the topic.

\bibliographystyle{plain}
\bibliography{prabhanjan}

\end{document}